\documentclass{llncs}
\usepackage{graphics}
\usepackage{graphicx}
\usepackage{amsmath}
\usepackage{amssymb}
\usepackage{algorithm}
\usepackage{algorithmic}
\usepackage{mathdots}

\usepackage{tikz}
\usetikzlibrary{positioning}
\usetikzlibrary{calc,trees,positioning,arrows,chains,shapes.geometric,%
    decorations.pathreplacing,decorations.pathmorphing,shapes,%
    matrix,shapes.symbols,automata}

\usetikzlibrary{fit}
\usetikzlibrary{shapes}
\usetikzlibrary{decorations.markings}

\tikzstyle{randomVariable}=[circle,fill=white,draw=black,text=black,minimum size=0.8cm]
\tikzstyle{param}=[draw=none,fill=none,text=black,minimum size=0.2cm]

\newtheorem{defin}     {Definition}
\newtheorem{theo}      {Theorem}

\newtheorem{cor}  {Corollary}
\newtheorem{prop}      {Property}

\newtheorem{rem}       {Remark}
\newtheorem{exa}  {Example}

\begin{document}

\mainmatter  

\title{Strong aggregation of the Markov chains associated with   matching models 
based on the automorphism group of their compatibility graphs}

\author{Moyi Yang$^1$, Jean-Michel Fourneau$^{1,2}$}

\institute{
$1$  DAVID, Univ. Paris-Saclay, UVSQ, Versailles France\\
\email{moyi.yang@ens.uvsq.fr \\}
$2$, INRIA, Paris, France \\
\email{Jean-Michel.Fourneau@uvsq.fr\\}
 }
\maketitle
\thispagestyle{plain}
\begin{abstract}
We extend the analysis of strong aggregation to general compatibility graphs, 
focusing on item counts rather than positions, and exploring generalized greedy matching disciplines. 
We prove that under a condition of automorphism-based transition consistency, the associated Markov chain is strongly aggregable for an arbitrary graph with a non-trivial automorphism group. 
Furthermore, we extend our analysis to non-greedy matching disciplines, distinguishing scenarios where compatible items can or cannot coexist within the same state. This result is illustrated with a simple compatibility graph with 
a rich automorphism structure: the odd rings. 
For all scenarios, we investigate the strong aggregation properties of the resulting Markov chains. 
This work enhances the theoretical understanding of lumpability in stochastic matching models and provides a foundation for analyzing complex graph structures.
\end{abstract}

\section{Introduction}

We consider a matching system with random arrivals of items of various 
types in discrete time, as proposed in  \cite{MBMa21} and  \cite{MaMo18}. 
The number of types is finite. The items wait in 
their queue (one for each type) until they are matched. Once two 
items (one of each type)  are matched, both items leave the queues immediately. 
The possible matching between two types of item is described by:
\begin{itemize}
\item  
a compatibility graph  $G=(V,E)$, where $V$ is the set of types (or nodes)
and an edge $(i,j)$ is in $E$ if item of type $i$ can be matched with items of type $j$.
\item a matching discipline $\Phi$, which is needed when several possible matchings occur due to an arrival. 
When an arriving item is compatible with several types of items already waiting, it provides 
a couple of items that are selected to leave the system. Usual matching disciplines are First Come First Match (FCFM)
and Match the Longest (ML). 
\end{itemize}

Arrivals are described by a probability distribution $(\alpha_i)_{i\in V}$.
At each time slot, one can observe the arrival of an item $i$ with probability $\alpha_i$, or no arrival at all with 
probability $\alpha_\emptyset$, so that $\alpha_\emptyset+\sum_{i\in V}\alpha_i=1$.
Assuming that the discipline only takes into account the number of items of each type and possibly their order
of arrival, such a model is associated with a Discrete Time Markov Chain (DTMC in the following). 
Note that it is also possible to study such a system in continuous time, 
assuming that the arrivals of items follow independent Poisson processes.
Here we only consider discrete time, but our results can be 
applied to the continuous time models 
as they are obviously uniformisable.

We consider some matching disciplines which, to the best of our knowledge, have not been studied so far, a notable exception being the ML discipline.
The set of disciplines we consider is based on the number of items of each type and does not take into account
the order of arrivals. Therefore, it is not needed to represent states 
as "words" and items as  "letters", as in \cite{MaMo18}. Instead, to represent the state of the system, we simply use the number of items (or letters when we want to use 
some results based on the "words" representation). The simplest discipline of that kind is the RANDOM 
discipline, where the compatible item is selected uniformly at random. 
Even if this discipline is very simple and has strong similarity with the Processor Sharing discipline in queuing theory, 
the associated Markov chain is in general not reversible, as we prove in Section~2. 

Here we will prove that the Markov chain associated to the RANDOM matching discipline is 
strongly aggregable for partitions described with the automorphism group of the compatibility graph.  
We generalize the results on RANDOM in various directions: we first consider greedy discipline, where, intuitively, 
the matching always occurs when two compatible items are present in the system. Then we study 
two families of non-greedy disciplines: in the first case, the arriving item is rejected even if it matches
another item already present, while in the second case, we assume that the selection for the 
matching and the deletion is conditioned by a threshold. 
In every case, we found necessary conditions for strong aggregation. 
These two cases have to be studied separately because this modification of the matching mechanism has
a strong implication on the state-space. For greedy disciplines or disciplines with REJECTION, the 
states are related to Independent States (IS in the following) 
of the compatibility graph while such a relation does not hold for disciplines with THRESHOLD.  

\begin{defin} [Independent Set]
Let $G=(V,E)$ be a compatibility graph.
A subset $I\subseteq V$ is an independent set of $G$ if for all $x,y\in I$ with $x\neq y$, we have $(x,y)\notin E$.
The set of independent sets of $G$ is denoted by $\mathcal{IS}(G)$.
\end{defin} 

\subsubsection{Related work.}
Several results about strong aggregation of discrete-time Markov chains associated with high-level models are already known. Indeed, one can expect that the partition of the state-space needed to prove strong aggregation is easier to find in a high-level model. 
In \cite{GDFo03} and  \cite{BBFP04}, such results were obtained for Stochastic Automata Networks, where the Markov chains are obtained through Kronecker products of the matrices associated with the stochastic automata. 
Local properties of the automata \cite{GDFo03} or duplication of automata in the definition of the net \cite{BBFP04} provide sufficient structural conditions for aggregation.  
In \cite{DuHa89,BDHI05},  local symmetries in Stochastic Petri Nets also lead to sufficient conditions for strong aggregation of the 
state-space of the associated Markov chain. 
In \cite{FoYa24}, we considered an extension of the matching model where the compatibility between items is described by a multigraph, i.e., a graph where all the nodes have a self-loop. 
In these models, an item is compatible with another item of the same type. 
Thus, two items of the same type cannot exist within the system, and the two items of the same type may be matched and leave the system together. 
We considered the RANDOM matching discipline. 
Such a model is used
to represent selection among online gamers: the type of an item is the level (for instance ELO points) of the opponent a gamer wants to play against.  
If all the nodes have a self-loop, the Markov chain associated with the model is finite and ergodic under some simple assumptions on the arrival probabilities and on the compatibility graph. 
The chain is finite, and the states are exactly the independent sets of the graph. 
As the number of independent sets of $G$ may grow exponentially with the number of nodes in $G$, using strong aggregation to reduce the number of states and perform some numerical analysis is an efficient way to deal with the state-space explosion. 
Note that the multigraph compatibility model with FCFM discipline was also independently studied in \cite{BMMR20}, under a more general assumption that some nodes, but not necessarily all of them, carry a self-loop. 
Due to the FCFM discipline, a product form solution was proved for the steady-state distribution. 

Another approach based on strong aggregation was used in \cite{BCDF24} to compute the steady-state distribution of the model whose compatibility graph is the complete graph with four nodes minus one edge. 
The following argument was used to prove that the chain can be aggregated into the chain associated with the complete graph with three nodes. We begin with a classical notation. 

\begin{defin}[Neighbourhood]
Let $G=(V,E)$ be a graph, and let $x$ be a node of $V$. 
The neighborhood $\Gamma(x)$ of $x$ is the subset of $V$ defined by $\Gamma(x)=\left\{y\in V:(x,y)\in E\right\}$.  
\end{defin}

\begin{prop} 
Let $\mathcal{S}$ be the state-space and let $G$ be a compatibility graph with two nodes $x$ and $y$, such that $\Gamma(x) = \Gamma(y)$ 
and $y \notin \Gamma(x)$. Thus, $x$ and $y$ can be considered as twins. 
Let $m[x]$, (resp. $m[y]$) denote the number of items at state $m$. 
Let $\cal{M}$ be the discrete time 
Markov chain associated with $G$, FCFM discipline, and any arrival rate vector which implies the ergodicity of the chain. Then $\cal{M}$ is strongly aggregable for the partition ${\cal B}_0$,..,${\cal B}_k$ with
\begin{equation} \label{twins} 
{\cal B}_k = \{ m \in {\cal S}, ~m[x]+m[y]= k \}. 
\end{equation} 
Clearly, the states in ${\cal B}_m$ are obtained by merging $x$ and $y$ in the description of $m$. 
Furthermore, the aggregated chain is associated with the model with compatibility graph $G'$ where
nodes $x$ and $y$ have been merged. 
\end{prop} 
 
The utilization of the automorphism group to reduce the complexity of a Markovian model has been proposed before. 
In \cite{SMKi11}, lumping techniques are used to reduce the number of  ODEs to be solved in an SIS model. 
In \cite{KADK19}, the steady-state distribution is approximated through a lumping based on local symmetries. 
To the best of our knowledge, the application of graph automorphisms to the exact analysis of matching models has not been studied, except in the two references given above. 

Matching Models were originally motivated by the analysis of kidney exchanges \cite{KidneySite,Unve10}. 
A kidney exchange arises when a healthy person who wishes to donate a kidney is not compatible (blood types or tissue types) with the receiver. Two incompatible pairs, or possibly more, can then form a cyclic exchange so that each patient can receive a kidney from a compatible donor. 

\subsubsection{Organization of the paper.}
In Section~2, we give some stability conditions for greedy 
disciplines \cite{MaMo18} which are formally defined. 
We also show that the model associated with RANDOM discipline is not reversible. 
In Section~3, we present our first result on a simple topology: the odd rings 
(note that even rings are not stable, as they are bipartite graphs) with the RANDOM discipline. Then we extend the results to an arbitrary greedy discipline. 
Section~4 is devoted to the generalization to an arbitrary graph with a non-trivial automorphism group and to an arbitrary greedy discipline. We illustrate the results with several
examples. 
In Section~5, we consider more complex disciplines which are not greedy. 
First, we study disciplines based on rejection, then the discipline based on a threshold.
Section~6 concludes the paper.

\section{Notation and Assumptions}

\subsection{Compatibility graph and states}
Let $G=(V,E)$ be the compatibility graph.
We use $n$ to denote the cardinality of $V$, implying the $n$ distinct types of items.
The set $E$ comprises edges $(x,y)\in V\times V$, where $(x,y)$ indicates compatibility between the item $x$ and $y$, and they are instantaneously removed from the system if selected by the matching discipline (in Section 5.3, this assumption will be modified). 
Without loss of generality, we assume that $G$ is connected and all the nodes in $G$ have no self-loops.

As the disciplines we consider only depend on the number of items of each type, we represent the state of the system by a vector of these numbers.
More precisely, a state is a vector $m=\left(m[1],\cdots m[n]\right)\in\mathbb{N}^n$, where $m[u]$ is the number of items of type $u$ waiting in the system.
We use the following notations:
\begin{itemize}
    \item $e_u$ is the vector of size $n$ whose entries are all equal to $0$, except the entry $u$ which is equal to $1$. The operator "$+$" is the component-wise addition of vectors;
    \item $|m|=\sum_{u\in V}m[u]$ is the total number of items waiting in state $m$;
    \item $\operatorname{supp}(m)=\{u\in V:m[u]>0\}$ is the support of $m$, which is the set of types which are present in state $m$;
    \item $\Gamma(m)=\Gamma\left(\operatorname{supp}(m)\right)$ is the set of the types which are compatible with at least one item waiting in state $m$. An arriving item of type $x$ is said to be compatible with state $m$ when $x\in\Gamma(m)$;
    \item $\Delta(m,x)=\Gamma(x)\cap\operatorname{supp}(m)$ is the set of the types which are present in $m$ and compatible with $x$.
\end{itemize}

\subsection{Matching disciplines}
\begin{defin}[Matching Discipline]
    Let $m$ be an arbitrary state, and $x$ be any item arriving into the system such that $x\in\Gamma(m)$. The matching discipline $D$ can be defined as a function mapping from $\mathcal{S}\times V$ to a non-empty subset $R(m,x)$, together with a probability distribution $\mu_{R(m,x)}$ on $R(m,x)$. 
    The state set $R(m,x)$ contains all possible successor states which are provoked by the arrival of item $x$ in state $m$. 
    And $\mu_{R(m,x)}(r)$ denotes the transition probability 
    from state $m$ to any state $r$ in $R(m,x)$ upon the arrival of item $x$.
    Furthermore, the distribution satisfies the usual condition:
    $$\sum_{r\in R(m,x)}\mu_{R(m,x)}(r)=1.$$
    \label{def: matching discipline}
\end{defin}
It is worth mentioning that the greedy matching discipline is a strategy where each arriving item $x$ must immediately match a compatible item that exists.
This implies that the greedy matching disciplines focus on making the best available decision at the current moment, without deferring to future options or considering long-term outcomes.

\begin{defin}[Greedy Matching Discipline]
    Let $m$ be an arbitrary state and $x$ be any item arriving in the system.
    If for all states $r\in R(m,x)$, we have $|r|=|m|-1$ and $x\notin r$, then the matching discipline is considered as a greedy matching discipline, denoted by $GD$.
    \label{def: greedy matching discipline}
\end{defin}
Typical greedy matching disciplines are First Come First Match, Match the Longest, RANDOM, and Priority.

\begin{defin}[RANDOM Discipline]
Under the RANDOM matching discipline, the item that is matched with the arriving item is chosen uniformly at random among all the items waiting in the system which are compatible with it.
Therefore, for $x\in\Gamma(m)$ and $u\in\Delta(m,x)$,
\[\mu_{R(m,x)}(m-e_u)=\frac{m[u]}{\sum_{v\in\Delta(m,x)}m[v]}.\]
\end{defin}

\begin{defin}  [Priority Discipline]
\label{def:priority}  
    Let $\kappa:V\rightarrow\mathbb{N}$ be a priority function, a small value meaning a high priority.
    Under the priority discipline associated with $\kappa$m the arriving item $x$ is matched with an item whose type belongs to 
    \[\Delta^\kappa(m, x)=\arg \min _{u \in \Delta(m, x)} \kappa(u).\]
    When $\kappa$ is injective, the priority is said to be \emph{global} and the arriving item always matches the item of highest priority.
    When $\kappa$ is not injective, several types share the same priority level, and the priority is said to be \emph{partial}.
\end{defin}
We illustrate the above definitions with an example of a ring compatibility graph $C_n$ and the RANDOM matching discipline.

\begin{exa}
    Consider a ring compatibility graph $C_n$ with $n$ odd, and the RANDOM matching discipline. 
    The associated Markov chain is denoted by $\mathcal{M}\left(C_n,(\alpha_i)_{i=0,\cdots,n-1}, RANDOM\right)$.
    Let $m$ be an arbitrary state and let $x$ be an arriving item with $x\in\Gamma(m)$.
    Since every node of $C_n$ has exactly two neighbours, $\Delta(m,x)$ contains one or two types.
The arrival of $x$ provokes the transition to the state $r=m-e_j$ by deleting one item of a compatible type $j\in\Delta(m,x)$. 
    The transition probability is
    $$\mu_{R(m,x)}(r)=\frac{m[j]}{\sum_{k\in\Delta(m,x)}m[k]}.$$
\end{exa}

Our goal is to establish relations between the automorphisms of the compatibility graph and the lumping of the Markov chain. Before that, we recall known results about the stability of the DTMC involved. 

\subsection{Stability}
The results of this subsection come from \cite{MaMo18}.
\begin{defin}[$\operatorname{NCOND}$ condition]
Let $G=(V,E)$ a compatibility graph and let $\alpha_1,..,\alpha_n$ be the arrival probabilities of the items of type 
$1$ to $n$. 
The  condition $\operatorname{NCOND}(G)$ on $ (\alpha_1,..,\alpha_n)$ is defined by  
 \[
 \forall I \in  {\cal{IS}}, ~~ \alpha_{I} < \alpha_{\Gamma(I)} . 
 \]
\end{defin} 

\begin{prop}
$\operatorname{NCOND}(G)$ is a necessary condition for the stability of the associated Markov chain.  
\end{prop} 

\begin{prop}
If $G$ is bipartite, then $\operatorname{NCOND}(G) = \emptyset$. Consequently, the Markov chain associated with a 
bipartite compatibility graph is unstable for all arrival probability vectors.
\label{prop:bipartite}
\end{prop}
\begin{rem}
\label{rem:bipartite}    Property \ref{prop:bipartite} is the reason why we restrict ourselves to rings $C_n$ with $n$ odd in the illustrative examples of Section~3. The rings of even nodes are bipartite, hence unstable.
    Note that the lumpability results proved in this paper do not require the ergodicity of the chain. The stability is only needed if one wants to use the aggregated chain to compute a steady-state distribution.
\end{rem}

\subsection{Symmetry and non-reversibility}
Looking at the DTMC associated with RANDOM discipline, one 
can easily check that the directed graph associated with the chain is symmetrical
as shown below. 

\begin{proposition}
Let $m$ be a state and $u$ be a type. Assume that $u \notin \Gamma(m)$.  Then the item $u$ enters the system and is not matched. 
Thus, there is a transition to state $m + e_u$ with probability $\alpha_u$. 

Now from state $m + e_u$, any arrival of an item compatible with $u$ will provoke a transition 
to $m$. This transition has a probability $\alpha_{\Gamma(u)}$ multiplied by the probability
to select a type $u$ item among the compatible population. As this probability is positive, the transition exists, 
and the directed graph of the Markov chain is symmetric.
\end{proposition} 

Unfortunately, the symmetry of the chain does not imply the reversibility of the chain associated with the arbitrary graph in general.

\begin{theo} 
Let $G$ be an arbitrary connected compatibility graph with
two nodes $x$ and $y$ which are not connected. 
Then the Markov chain associated with $G$, the RANDOM matching  discipline and arrival probability vector
$(\alpha_1,...,\alpha_n)$ which insures ergodicity, 
is not reversible. 
\end{theo} 
\begin{proof}
Remember that if a DTMC is ergodic and reversible, with stochastic matrices $Q$, 
 then for any state $m$, the product 
of ratio of probabilities $\frac{Q(m_i,m_{i+1})}{Q(m_{i+1},m_i)}$ among a path from $x$ to the empty state
is constant for all direct paths. 

Now, consider state $m$, such that $m[x]=1$, $m[y] = 2$ and $m[u]=0$ for all other types
of item. Remark that this state exists as nodes $x$ and $y$ are not connected in $G$, there, the items do not match. 
Let $V1 = \Gamma(x) \cap \Gamma(y)$, 
$V2 = \Gamma(x) \setminus \Gamma(y)$, and $V3 = \Gamma(y) \setminus \Gamma(x)$.  
Let us now consider two paths from $m$ to $\emptyset$:
\begin{itemize}
\item Path  $m$, $m_1$ with  $m_1[x]=0$, $m_1[y] = 2$ and $m_1[u]=0$, 
$m_2$ with  $m_2[x]=0$, $m_2 [y] = 1$ and $m_2[u]=0$ and finally $\emptyset$: the ratio of
probability is $$\frac{\alpha_x \alpha_y^2}{(\alpha_{V1} *1/3 + \alpha_{V2}) (\alpha_{V1}+\alpha_{V3})^2}$$.
The coefficient $1/3$ is the ratio of probability to select the type $x$ item among the three compatible items. 

 \item Path  $m$, $m_3$ with  $m_3 [x]=1$, $m_3 [y] = 1$ and $m_3[u]=0$, 
$m_4$ with  $m_4 [x]=1$, $m_4 [y] = 0$ and $m_4[u]=0$ and finally $\emptyset$: 
 the ratio of
probability is $$\frac{\alpha_x \alpha_y^2}{(\alpha_{V1} *2/3 + \alpha_{V3}) (\alpha_{V1} *1/2 + \alpha_{V3}) (\alpha_{V1}+\alpha_{V2})}$$.

\end{itemize} 
In general, without further assumptions on the compatibility graph, the denominators are not equal and the chain is not reversible. 
\end{proof}

This property also holds when the compatibility relation is modeled by a multigraph (we add some 
self-loops in the graph).

\begin{cor} 
Let $H$ be an arbitrary connected 
compatibility multigraph (i.e. a graph where some nodes have a self-loop) 
with
two nodes $x$ and $y$ which are not connected. Note that if all the nodes have a self loop, the state space is finite. 
Then, a similar proof shows that the Markov chain associated with $H$, the RANDOM matching  discipline and arrival probability vector
$(\alpha_1,...,\alpha_n)$ 
is also not reversible. 
\end{cor} 
We are still looking for sufficient conditions on the compatibility graph and the arrival rates to prove the reversibility of these Markov chains.

\subsection{Strong aggregation and lumpability}
Let us now recall the definitions of strong aggregation and ordinary lumpability.

\begin{defin}[Strong Aggregation]
Let $W$ be a Markov chain on the state-space $\mathcal{S}$. Consider a partition $(B_1,..,B_k)$ of $\mathcal{S}$. We define the aggregated process  
 $Y$ as follows:
 \[
 Y_n = i \quad \Longleftrightarrow\quad  W_n \in B_i. 
 \]
We aim to find conditions under which $Y$ also forms a Markov chain. We denote $Y$ as a strong aggregation of $W$ for 
partition $(B_1,..,B_k)$. 
\end{defin}

\begin{defin}[Ordinary Lumpability] 
\label{lump} 
W is strongly lumpable for partition $(B_1,..,B_k)$ of its state-space if, for all subset indices $i$ and $j$, and for all states $m_1$ 
and $m_2$ in $B_i$, the following condition holds: 
\[
Pr(W_{n+1} \in B_j \vert W_n = m_1) = Pr(W_{n+1} \in B_j \vert W_n = m_2). 
\]
Here, $B_i$ is denoted as macro-state $i$. Consequently, the blocks of the transition matrix associated with the macro-states of the partition exhibit a constant row sum. 
\end{defin}

The primary challenge lies in finding the appropriate partition and verifying the lumpability condition. 
We first present a simpler example before delving into the more general and formal results.

\section{The ring compatibility graph $C_n$ with $n$ odd}

We consider the ring graph $G=C_n$ with $n$ nodes, $n$ being odd for the sake of stability (recall Remark \ref{rem:bipartite}).
In Figure \ref{fig:C_5}, we have drawn a part of the Markov chain associated with a ring of size $5$ and the RANDOM discipline. 
 We only represent states with at most three items. The states and the loops on each node associated with probability $\alpha_\emptyset$ are omitted for the sake of readability.  Clearly, the graph of the chain
 has a large number of symmetries and one may ask whether this property can be used to simplify the analysis. 
 
 \begin{figure} [hbtp] 
\begin{center}
 \includegraphics[scale=0.122] {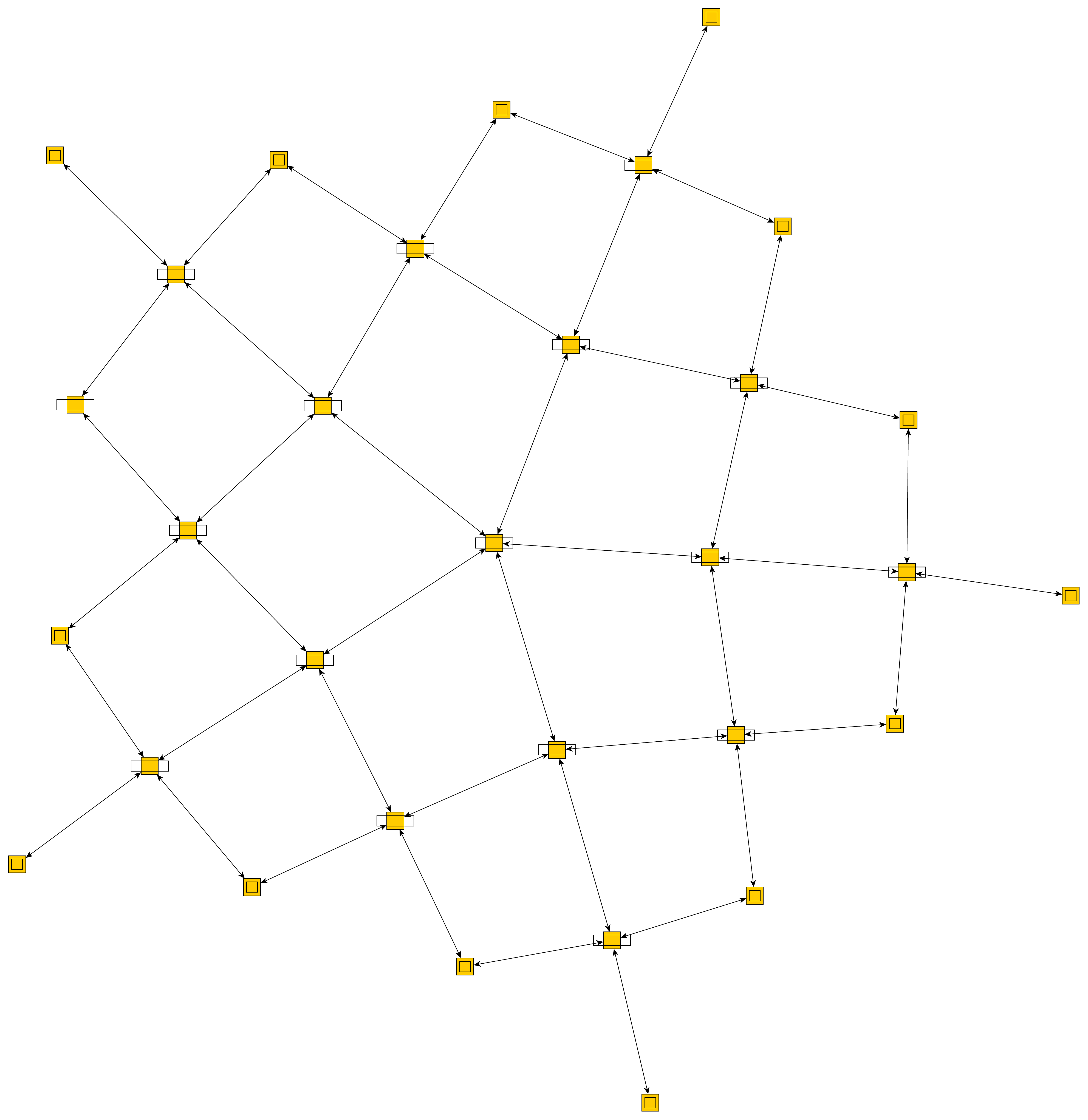} 
 \caption{Graph of the chain associated with $C_5$. }
 \end{center} 
 \label{fig:C_5}
 \end{figure} 
The nodes are labelled from $0$ to $n-1$ to simplify the definition of rotations.

\begin{note}[Operation $\oplus$]
    For $i,j\in\{0,\cdots,n-1\}$, we write $i\oplus j=(i+j)\mod n$
\end{note}
\begin{defin}[Rotations]
For $k\in \{0,\cdots,n-1\}$, the permutation $\sigma_k$ of the set $V$ of nodes of $C_n$ is defined by $\sigma_k(i) = i \oplus k$. Clearly, $\sigma_k ^{-1} (i) = i \oplus (n- k) $. 
Recall that the set of rotations forms a subgroup of the permutations.
\end{defin}

\begin{note} 
The permutation $\Lambda_k  $ associated with $\sigma_k  $  is a 
permutation on the state-space $\mathcal{S}$ of the Markov chain, defined as $\Lambda_k (m) = 
(m_{\sigma_k^{-1} (0)}, ..., m_{\sigma_k^{-1} (n-1)})$.  
\end{note} 

We consider the following partition of the state-space, and we prove the lumpabillity of the Markov chain 
for this partition when certain constraints on the arrival distribution and on the discipline are satisfied. 
\begin{equation} \label{part1} 
{\cal B}_m = \{ p \in {\cal S} ~\vert ~\exists k, ~p=\Lambda_k (m) \}. 
\end{equation} 
Clearly, the states of ${\cal B}_m$ are obtained by rotating the description of $m$. 
Since each node in the compatibility 
graph has no self-loop, the state-space is infinite.
However, the number of transitions out of an arbitrary state is finite.
More precisely, since the graph is a ring, each node has only two neighbours.
The number of possible transitions provoked by arrivals is relatively small: one inclusion, one deterministic deletion, or one deletion among two types of items. The proofs below are often based on this decomposition of the transitions. 

\begin{lemma} \label{lemma:cycle} 
Let $m$ be an arbitrary state and $x$ an item type. Let $p$ be a state in the same macro-state as $m$. 
Therefore, there exists $k$ such that $p=\Lambda_k (m)$. Let $y = \sigma_k (x) $ ($\sigma_k$ is the rotation  associated 
with the same $k$). Based on the type of transition, one can state two properties under arbitrary greedy matching disciplines: 
\begin{itemize} 
\item 
 If the action of $x$ on $m$ leads to a state $r$, then the action of $y$ on $p$ leads to a state $s$. Furthermore, $r$ and  $s$ are in the same macro-state. 
\item If the action of $x$ on $m$ leads to a state $r$ drawn among $\{r_1, r_2\}, $
then the action of $y$ on $p$ leads to a state $s$ among $\{s_1, s_2\}$. 
Furthermore, $r_1$ and  $s_1$ are in the same macro-state, and a similar statement 
holds for $r_2$ and  $s_2$. 
\end{itemize} 
\end{lemma} 
To formalize the analysis, we identify three distinct cases which describe the effect of the arrival of an item of type $x$ on the state $m$.
A detailed proof is given for each case.
\begin{enumerate}
    \item Case 1: $\Gamma(x)\cap \operatorname{supp}(m)=\emptyset$. With probability 1, the item $x$ is added to the state $m$, which becomes $m+e_x$.
    \item Case 2: $\Delta(m,x)\neq\emptyset$  and $|\Delta(m,x)|=1$.
    In this case, we assume that $z$ is the type.
    With probability 1, the resulting state could be $r=m-e_z$, obtained by deleting one item of type $z$.
    \item Case 3: $\Delta(m,x)\neq\emptyset$ and $|\Delta(m,x)|=2$.
    In this case, we assume that $u$ and $z$ are the two neighbours of $x$.
    Thus, the resulting state could be $r_1=m-e_u$ obtained by deleting item $u$ with probability $\mu_{R(m,x)}(r_1)$, or state $r_2=m-e_z$ obtained by deleting item $z$ with probability $\mu_{R(m,x)}(r_2)$.
\end{enumerate}
Note that the first case does not depend on the graph $G$.
The last two cases take into account that all the nodes of $G$ have two neighbors.
\begin{proof}
The proof can be organized following the three cases:
\begin{itemize}
    \item Case 1: In the ring compatibility graph, if $\Gamma(x)\cap \operatorname{supp}(m)=\emptyset$, this implies that $m[x\oplus 1]=0$ and $m[x\oplus(n-1)]=0$.
    First, we must demonstrate that the effect of $y$ on $p$ also results in a new item arriving in $p$, by showing that
    $p[y\oplus 1]=0$ and $p[y\oplus(n-1)]=0$.

    Given that $p=\Lambda_k(m)$, we have $p[i]=m[\sigma_{k}^{-1}(I)]=m[i\oplus(n-k)]$ for all $i$.
    Since $y=x\oplus k$, we can compute:
    $$p[y\oplus 1]=m[x\oplus k\oplus(n-k)\oplus 1]=m[x\oplus 1]=0,$$ and 
    $$p[y\oplus (n-1)]=m[x\oplus k\oplus(n-k)\oplus (n-1)]=m[x\oplus (n-1)]=0.$$
    Therefore, the arrival of $y$ in state $p$ also results in the arrival of a new item.

    Next, we verify that these states belong to the same macro-state.
    By assumption, both states $m$ and $p$ belong to the same macro-state $\mathcal{B}_m$.
    Therefore, there exists some $k$ such that $p=\Lambda_k(m)$.
    Let $r$ denote the state resulting from the arrival of $x$ in state $m$, and $s$ denote the state resulting from the arrival of $y=\sigma_k(x)$ in $p=\Lambda_k(m)$.
    We need to prove $s\in\mathcal{B}_r$. It was previously shown that the action of $x$ on $m$ leads to $r=m+e_x$.
    Similarly, the action of $y$ on $p$ results in state $s=p+e_y$. Furthermore, $$\Lambda_k(r)=\Lambda_k(m+e_x)=\Lambda_k(m)+e_{\sigma_k(x)}=p+e_y=s.$$
    Therefore, state $s$ belongs to partition $\mathcal{B}_r$, concluding Case 1.
    
    \item Case 2: In the case where $\Delta(m,x)\neq\emptyset$ and $|\Delta(m,x)|=1$, we assume $\Delta(m,x)=\{z\}$.
    Without loss of generality, let $z=x\oplus 1$(the other case is $z=x\oplus (n-1)$).
    The effect of item $x$ is a transition to state $r=m-e_{z}$, which gives $m[x\oplus 1]=m[z]$ and $m[x\oplus (n-1)]=0$. Consequently, we have $r=m-e_{x\oplus 1}$.
    Now, consider $p=\Lambda_k(m)$ and the arrival of item $y=\sigma_k(x)$ at state $p$. We have:
    $$p[y\oplus 1]=p[x\oplus k\oplus 1]=m[x\oplus k\oplus 1\oplus (n-k)]=m[x\oplus 1]=m[z],$$
    and $$p[y\oplus (n-1)]=p[x\oplus k\oplus (n-1)]=m[x\oplus (n-1)]=0.$$
    Therefore, the arrival of item $y$ provokes the transition from state $p$ to state $s=p-e_{x\oplus k\oplus 1}$ with probability 1.
    Furthermore, $$\Lambda_k(r)=\Lambda_k(m-e_{x\oplus 1})=\Lambda_k(m)-e_{x\oplus k\oplus 1}.$$
    Thus, we have $\Lambda_k(r)=s$, and $r$ and $s$ are in the same macro-state.
    
    \item Case 3: In the case where $\Delta(m,x)\neq\emptyset$ and $|\Delta(m,x)|=2$, we assume that $\Delta(m,x)=\{u,z\}$.
    The resulting states could be state $r_1=m-e_{u}$ obtained by deleting item $u$ with probability $\mu_{R(m,x)}(r_1)$, or state $r_2=m-e_{z}$ obtained by deleting item $z$ with probability $\mu_{R(m,x)}(r_2)$.
    Without loss of generality, let $u=x\oplus 1$ and $z=x\oplus(n-1)$. Therefore, $m[x\oplus 1]=m[u]$ and $m[x\oplus (n-1)]=m[z]$.
    Let $p=\Lambda_k(m)$ and $y=\sigma_k(x)$.
    Hence, we have $$p[y\oplus 1]=p[x\oplus k\oplus 1]=m[x\oplus k\oplus 1\oplus(n-k)]=m[x\oplus 1]=m[u]$$,
    and $$p[y\oplus (n-1)]=p[x\oplus k\oplus (n-1)]=m[x\oplus (n-1)]=m[z].$$
    The arrival of item $y$ provokes the deletion of item $y\oplus 1$ with probability $\mu_{R(p,y)}(s_1)$, leading to state $s_1$ or item $y\oplus (n-1)$ with probability $\mu_{R(p,y)}(s_2)$, leading to state $s_2$.
    The same arguments as in Case 2 prove that $r_1$ and $s_1$ are in the same macro-state. Similarly, we have the same result for states $r_2$ and $s_2$.
\end{itemize}
\end{proof}

Lemma \ref{lemma:cycle} is a purely structural statement. 
It says that the sets of possible successor states correspond to each other under the rotation, but nothing about the probabilities of these successor states. 
In Cases 1 and 2 the transition probabilities are equal to $1$ because there is no ambiguity, but in Case 3 the arrival of $x$ provokes a transition to $r_1$ (or $r_2$) with probability $\mu_{R(m,x)}(r_1)$ (or $\mu_{R(m,x)}(r_2)$), while the arrival of $y$ provokes a transition to $s_1 = \Lambda_k(r_1)$ (or $s_2 = \Lambda_k(r_2)$) with probability $\mu_{R(p,y)}(s_1)$ (or $\mu_{R(p,y)}(s_2)$), and these two distributions may differ. 
We therefore add to Lemma \ref{lemma:cycle} two invariance assumptions, one on the arrival probabilities and one on the discipline, which together turn the correspondence between the successor states into an equality between the transition probabilities.

\begin{cor}
    Adopt the setting and the notation of Lemma \ref{lemma:cycle}. Assume that
\begin{enumerate}
    \item the arrival probabilities are invariant under the rotations:
    \[\alpha_x = \alpha_{\sigma_k(x)} \quad \forall x \in V\]
    \item the discipline is invariant under the rotations: 
    \begin{equation}
        \mu_{R(m,x)}(r) = \mu_{R(p,y)}(s) \quad \forall m, r \in\mathcal{S},x\in  V.
        \label{eq:equiv_mu}
    \end{equation}
\end{enumerate}
Then the two corresponding transitions have the same probability,
\begin{equation}
    \label{eq:equiv_alpha_mu}\alpha_x\mu_{R(m,x)}(r) = \alpha_y \mu_{R(p,y)}(s).
\end{equation}
\label{cor:conditions}
\end{cor}

\begin{exa}
\label{exa:rnd}
    The condition in \eqref{eq:equiv_mu} is naturally satisfied under certain matching disciplines. 
    For example, in the case of the RANDOM matching discipline previously defined.
    As demonstrated in the proof of Case 3 in Lemma \ref{lemma:cycle}, the arrival of item x provokes the deletion of item $x\oplus 1$ with probability according to the RANDOM matching discipline:
    $$\mu_{R(m,x)}(r_1)=\frac{m[x\oplus 1]}{m[x\oplus 1]+m[x\oplus (n-1)]}.$$
    Based on the assumption $p=\Lambda_k(m)$, we have proved the following equation:
    \begin{equation}
    \label{eq:exa}
        p[y\oplus 1]=m[x\oplus 1],\quad \text{and}\quad p[y\oplus (n-1)]=m[x\oplus(n-1)]
    \end{equation}
    In this scenario, the arrival of item $y$ provokes the deletion of item $y\oplus 1$ with probability based on \eqref{eq:exa}
    $$\mu_{R(p,y)}(s_1)=\frac{p[y\oplus 1]}{p[y\oplus 1]+p[y\oplus (n-1)]}.$$
    resulting in a transition to state $s_1$.
    Therefore, we could derive the following equation in this scenario:
    $$\mu_{R(m,x)}(r_1)=\mu_{R(p,y)}(s_1).$$
    Similarly, for the state $s_2=\Lambda_k(r_2)$, we have
    $$\mu_{R(m,x)}(r_2)=\mu_{R(p,y)}(s_2).$$
    Clearly, under the assumption that $\alpha_x=\alpha_y$, the transition probability satisfies, for any $s=\Lambda_k(r)$, the required condition holds:
    $$\alpha_x\mu_{R(m,x)}(r)=\alpha_y\mu_{R(p,y)}(s)$$
\end{exa}

\begin{exa}
    \label{exa:priority}
    We proved in Example \ref{exa:rnd} that we can satisfy the relevant conditions concerning the transition probability $\mu$ under the RANDOM matching discipline to further exemplify that the lumpability condition holds.
    In contrast, there also exist greedy matching disciplines for which the condition on $\mu$ fails even when $\alpha_i=\alpha$ for all $i$.
    We take the priority matching discipline of Definition \ref{def:priority} as an example.
    When the arrival of an item may provoke several matchings, the item is matched with the item of highest priority.
    Consider the ring $G=C_7$ and the global priority
    \[\kappa:0>1>2>3>4>5>6\]
    Consider the state $m=e_1+e_3$ and the arriving item $x=2$.
    We can obviously analyze this case under Case 3 of Lemma \ref{lemma:cycle} since $\Delta(m,x)=\{1,3\}$.
    Take $k=4$, so that $p=\Lambda_4(m)=e_5+e_0$ and $y=\sigma_4(2)=6$.
    The arrival of $x=2$ may provoke the deletion of the item $x\oplus 1=3$ leading to state $r_1=m-e_{3}$, or the deletion of item $x\oplus (n-1)=1$ leading to state $r_2=m-e_{1}$.
    As the priority of type $1$ is higher than the priority of type $3$, the arrival of $x=2$ deletes the item of type $1$:
    $$\mu_{R(m,x)}(r_1)=0,\quad \mu_{R(m,x)}(r_2)=1.$$

    Meanwhile, the arrival of item $y=\sigma_k(x)=6$ may provoke the deletion of the item $y\oplus 1=0$, leading to state $s_1=p-e_{0}=\Lambda_4(r_1)$, or the deletion of the item $y\oplus (n-1)=5$, leading to state $s_2=p-e_{5}=\Lambda_4(r_2)$.
    As the priority of type $0$ is higher than the priority of type $5$,
    $$\mu_{R(p,y)}(s_1)=1,\quad \mu_{R(p,y)}(s_2)=0.$$
    Obviously, even if $\alpha_i=\alpha$ for all item types, we still cannot satisfy the condition in \eqref{eq:equiv_alpha_mu}
\end{exa}
\begin{exa}
    \label{exa: partial priority}
    
    We now consider a partial priority matching discipline, where some types share the same priority level.
    Such a discipline assigns local priorities based on the structure of the underlying graph.
    To illustrate this, we consider the automorphism group of the cycle graph $C_7$, specifically the dihedral group $D_7$.
    Unlike previous examples, which were based on the rotation subgroup, we now consider the reflection subgroup.
    Each reflection fixes one node while swapping symmetric pairs of nodes.

    As a concrete example, we take the reflection permutation 
    $$\sigma_0=(0),(1,6),(2,5),(3,4),$$
    which fixes node $0$ and swaps the remaining nodes symmetrically.
    In this setting, we define a partial priority:
    $$\kappa:\{0\}>\{1,6\}>\{2,5\}>\{3,4\}.$$
    Consider an arbitrary state $m$, an arriving item $x$. 
    Let $p=\Lambda_{\sigma_0}(m)$ with an arriving item $y=\sigma_0(x)$, under the same reflection permutation $\sigma_0$.
    In this example, the partial priority matching discipline $\kappa$ will affect the results when $|\Delta(m,x)|=2$, so we will focus on this case.
Let $\Delta(m,x)=\{z_1,z_2\}$, so that the arrival of $x$ leads to $r_1=m-e_{z_1}$, or the deletion of the item $z_2$ leading to state $r_2=m-e_{z_2}$.
    Similarly, for the state $p=\Lambda_{\sigma_0}(m)$, the arrival of item $y=\sigma_0(x)$ would provoke the deletion of the item $\sigma_0(z_1)$ leading to state $s_1=p-e_{\sigma_0(z_1)}$, or the deletion of the item $\sigma_0(z_2)$ leading to state $s_2=p-e_{\sigma_0(z_2)}$.
Since $\kappa$ is constant on the orbits, we have:
    \begin{equation}
        \label{eq:partial priority}
        \kappa(z_1)=\kappa(\sigma(z_1)),\quad \kappa(z_2)=\kappa(\sigma(z_2))
    \end{equation}
Consequently, the condition in \eqref{eq:equiv_alpha_mu} hold as soon as $\alpha_x=\alpha_{\sigma_0(x)}$ for all $x$.
Note that the arrival probabilities do not need to be uniform here, they only have to be constant on the orbits of the subgroup of the automorphism group.
\end{exa}
    Thus, we demonstrate through two examples (\ref{exa:priority} and \ref{exa: partial priority}) of priority matching disciplines that when partial priority matching disciplines can obey invariance under permutations, the system is still lumpable even if matching disciplines involve priority rules for specific item types.

We can now establish our first result about strong aggregation of the DTMC. 
\begin{theo}\label{theo:cycle}
    Consider the ring $C_n$ with $n$ odd, an arbitrary greedy matching discipline $GD$, and uniform arrival probabilities, which means that for all $i$, we have $\alpha_i=\alpha$.
    If the discipline satisfies the condition in \eqref{eq:equiv_mu}, then the Markov chain $\mathcal{M}(C_n,(\alpha_i)_{i=0,\cdots,n-1},GD)$ is strongly aggregable for the partition defined in \eqref{part1}.
\end{theo}
\begin{proof}
To establish ordinary lumpability, we consider two macro-states defined by equation \eqref{part1}: $\mathcal{B}_i$ and $\mathcal{B}_j$, as well as two arbitrary states $m$ and $p$ within $\mathcal{B}_i$.
Without loss of generality, suppose there is a non-zero probability of transitioning from state $m$ to $\mathcal{B}_j$. 
By construction:
$$Pr(m,{\cal{B}}_j) = \sum_{l=0}^{n-1} \alpha_l \sum_{u \in {\cal{B}}_j } Pr({\rm arrival ~of~ letter} ~{\it l}~ {\rm provokes~
a~ transition~ from} ~{\it m}~ {\rm to} ~{\it u}).$$
From Lemma \ref{lemma:cycle} and Corollary \ref{cor:conditions}, we know that for any $k$, we have:
$$\begin{array} {l} 
\sum_{u \in {\cal{B}}_j } Pr({\rm arrival ~of~ letter} ~{\it l}~ {\rm provokes~
a~ transition~ from} ~{\it m}~ {\rm to} ~{\it u})  = \\
\\
~~~~ \sum_{u \in {\cal{B}}_j } Pr({\rm arrival ~of~ letter} ~{\it \sigma_k (l) }~ {\rm provokes~
a~ transition~ from} ~{\it \Lambda_k (m) }~ {\rm to} ~{\it u}). 
\end{array} $$
Since $\sigma_k$ is a one-to-one mapping and $\alpha_l=\alpha$ for all $l$, we conclude that the lumpability condition holds for the cycle graph.
\end{proof}

Let us now generalize the previous result to an arbitrary graph. For rings of odd size, the rotations and the reflections 
are the automorphisms of the compatibility graph. 
This is the tool we need for a more general compatibility graph. 

\section{An Arbitrary Graph with a Non-trivial Automorphism Group}

We begin by recalling the definition of the automorphism group of a graph.

\begin{defin}[Graph automorphism]
    Consider an arbitrary graph $G=(V,E)$ and let $\sigma$ be a permutation of the set $V$. 
    We say that $\sigma$ is a graph automorphism if, for all nodes $u$ and $v$, $(u,v)\in E$ implies $(\sigma(u),\sigma(v))\in E$.
\end{defin}
Assume that $G$ has a non-trivial automorphism group, denoted $\mathrm{Aut}[G]$, and that $V$ contains $n$ nodes, none of which carries a self-loop.
Let $\mathcal{A}$ represent a subgroup of $\mathrm{Aut[G]}$.
For example, in the previous section, we considered the automorphism group of a ring with $n$ nodes, which is the dihedral group, and we mainly restricted our attention to the subgroup of the rotations. 

Let $\mathcal{S}$ be the state space of the Markov chain associated with the compatibility graph $G$, the arrival probability distribution $\alpha_i$, and a greedy matching discipline.
For a permutation $\sigma\in\mathcal{A}$, define $\Lambda_\sigma$ as the corresponding permutation of the states $m=(m_1,\cdots,m_n)$ in $\mathcal{S}$ by
$$\Lambda_\sigma(m_1,\cdots,m_n)=(m_{\sigma^{-1}(1)},\cdots,m_{\sigma^{-1}(n)})$$
Since $\mathcal{A}$ is a subgroup of $\mathrm{Aut}[G]$, the inverce permutation $\sigma^{-1}$ also belongs to $\mathcal{A}$.
This setup leads to the following property, which generalizes previous propositions about rotations.

\begin{property}
    Let $m$ and $p$ be two states in $\mathcal{S}$. 
    For every automorphism $\sigma$ in $\mathcal{A}$, we have
    $$\Lambda_\sigma(m+p)=\Lambda_\sigma(m)+\Lambda_\sigma(p), \text{ and } \Lambda_\sigma(e_x)=e_{\sigma(x)}.$$
\end{property}
With this extension to arbitrary graphs, we introduce a natural partition of the state space as follows:
$$\mathcal{B}_m=\{p\in\mathcal{S}|\exists\sigma\in\mathcal{A},p=\Lambda_\sigma(m)\}.$$
Thus, each partition $\mathcal{B}_m$ consists of states obtained by applying any permutation from $\mathcal{A}$ to the description of $m$.
Since explicitly constructing these partitions for the transition matrix is complex, we instead aim to prove a stronger property that ensures lumpability.

\begin{lemma}\label{lemma:arbitrary}
    Let $m$ be an arbitrary state of $\mathcal{S}$, let $x$ be an item type of $V$, and let $\sigma$ be a permutation of the subgroup $\mathcal{A}$.
    Let $p=\Lambda_\sigma(m)$ and $y=\sigma(x)$.
    Under an arbitrary greedy matching discipline, if the action of $x$ on $m$ leads to a state $r$, then the action of $y$ on $p$ leads to a state $s$.
    Furthermore, $r$ and $s$ are in the same macro-state.
    
\end{lemma}

To proceed with the proof, we consider two cases based on the action of $x$ on $m$. Note that, unlike in Lemma \ref{lemma:cycle}, a node of $G$ may have an arbitrary number of neighbours.
\begin{enumerate}
    \item Case 1: $\Gamma(x)\cap \operatorname{supp}(m)=\emptyset$. 
    Here, the item $x$ is added to $m$ with probability $1$, yielding $m+e_x$ as the new state.
    \item Case 2: $W=\Delta(m,x)$ is not empty.
    Let $q=|W|$ and let $z_1,\cdots z_q$ be the types of $W$.
\end{enumerate}
\begin{proof}
The proof structure is analogous to that of Lemma \ref{lemma:cycle}.

\begin{itemize}
    \item Case 1: For an arbitrary compatibility graph without self-loops, the condition $\Delta(m,x)=\emptyset$ implies that for all $u\in\Gamma(x)$, we have $m[u]=0$.
    Let $p=\Lambda_\sigma(m)$ for some $\sigma\in\mathcal{A}$, and let $y=\sigma(x)$.
    By definition, we have for all $u\in\Gamma(x)$,
$$p[y]=p[\sigma(u)]=m[\sigma^{-1}(\sigma(u))]=m[u]=0.$$
    Since $\sigma$ is an automorphism, $(x,u)\in E$ if and only if $(\sigma(x),\sigma(u))\in E$.
    Therefore, for all $v\in\Gamma(\sigma(x))$, we have $p[v]=0$.
    This implies that the arrival of item $y$ in $p$ leads to state $p+e_y$ with probability 1.
    By the greedy matching discipline, neither $x$ nor $y$ can be discarded on arrival, ensuring that
$$\Lambda_\sigma(r)=\Lambda_\sigma(m+e_x)=\Lambda_\sigma(m)+e_{\sigma(x)}=p+e_y=s$$
    Therefore, state $s$ belongs to the partition $\mathcal{B}_r$, which completes Case 1.
    \item Case 2: In this scenario, the intersection $W=\Gamma(x)\cap\operatorname{supp}(m)$ is not empty, which means that there exists at least one type $z\in\Gamma(x)$ such that $m[z]\neq 0$.
    Since $\sigma$ is an automorphism of the graph $G$, it maps $z$ to $\sigma(z)$, ensuring that $\sigma(z)$ is a neighbour of $\sigma(x)$.
    Thus, we have:
    \begin{itemize}
        \item $z\in\Gamma(x)$ implies $(x,z)\in E$
        \item By the automorphism property, $(\sigma(x),\sigma(z))\in E$, thus $$\sigma(z)\in\Gamma(\sigma(x))\cap\Lambda_\sigma(m),$$ confirming that $\sigma(z)$ is a neighbour of $\sigma(x)$
        \item Since $z\in m$, we also have $m[z]\neq 0$, which implies:
        $$p[\sigma(z)]=m[\sigma^{-1}(\sigma(z))]=m[z]\neq 0$$
    \end{itemize}
    Because of the automorphism property, we also have    $$|\Gamma(\sigma(x))\cap\Lambda_\sigma(m)|=q\quad \text{for all }\sigma\in\mathcal{A}.$$
    The arrival of item $x$ in state $m$ or item $\sigma(x)$ in state $p=\Lambda_\sigma(m)$ results in a transition where one of the items in $\Gamma(x)\cap m$ is removed.
    
    To capture this transition, let $r_1,r_2,\cdots,r_q$ be the possible resulting states when item $x$ arrives at state $m$ and one of the $q$ neighbors $z_1,z_2,\cdots,z_q$ is removed.
    Correspondingly, let $s_1,s_2,\cdots, s_q$ be the states obtained by removing items $\sigma(z_1),\sigma(z_2),\cdots, \sigma(z_q)$ upon the arrival of $\sigma(x)$ in $p$.

    To show that each $r_i$ and $s_i$ belong to the same macro-state, we verify that: $s_i=\Lambda_\sigma(r_i)$ for each $i=1,\cdots,q$, ensuring that $r_i$ and $s_i$ are in the same partition $\mathcal{B}_{r_i}$.
    We define each $r_i$ by the formula: $r_i=m-e_{z_i}$, where $e_{z_i}$ represents the removal of item $z_i$ from state m.
    Applying $\Lambda_\sigma$ to $r_i$, we obtain:
    $$\Lambda_\sigma(r_i)=\Lambda_\sigma(m-e_{z_i})=\Lambda_\sigma(m)-e_{\sigma(z_i)}=p-e_{\sigma(z_i)}=s_i$$
    Thus, $\Lambda_\sigma(r_i)$ corresponds to the state obtained by removing $\sigma(z_i)$ from $p$, which matches $s_i$.
    Since $\sigma(z_i)$ is a neighbor of $\sigma(x)$, each $s_i$ belongs to the states resulting from the arrival of $\sigma(x)$ at $p$.

    \end{itemize}
\end{proof}

\begin{cor}
    \label{corollary:arbitrary}
    Adopt the setting and the notation of Lemma \ref{lemma:arbitrary}. Assume that
    \begin{enumerate}
        \item the arrival probabilities are invariant under the automorphism subgroup $\mathcal{A}$:
        \begin{equation}
            \alpha_x=\alpha_{\sigma(x)},\quad\forall x\in V,\forall \sigma\in\mathcal{A};
            \label{eq:arbitrary_alpha}
        \end{equation}
        \item the discipline is invariant under the automorphism subgroup $\mathcal{A}$:
        \begin{equation}
            \mu_{R(m,x)}(r)=\mu_{R(p,y)}(s),\quad \forall m,r\in\mathcal{S},\forall x\in V,\forall\sigma\in\mathcal{A}.
            \label{eq:arbitrary_mu}
        \end{equation}
        Then the two corresponding transitions have the same probability:
        \begin{equation}
            \alpha_x\mu_{R(m,x)}(r_i)=\alpha_y\mu_{R(p,y)}(s_i),\label{eq:remark}
        \end{equation}
    \end{enumerate}
\end{cor}
    
When the transition probabilities $\mu_{R(m,x)}(r_i)$ and $\mu_{R(p,y)}(s_i)$ are consistent except for a possible mismatch in arrival probabilities.
For example, consider a ring graph with uniform transition probabilities.
In this case, the transition probabilities depend only on the structure of the graph and the matching discipline.
If $\sigma$ is an automorphism of graph $G$, the transition probabilities for the arrival of $x$ under state $m$ and the arrival of $y=\sigma(x)$ under state $\Lambda_{\sigma}(m)$ are identical because the neighbor set $\Gamma(x)$ maps exactly to $\Gamma(y)$ under $\sigma$, and the matching probabilities are uniform.

It is obvious that the automorphism-based transition consistency constraint in \eqref{eq:remark} is indeed an infinite constraint.
Since $V$ is finite, the constraint does not technically involve an infinite set of nodes, but it is considered infinite in the sense that it applies universally across all possible instances generated by $\mathcal{A}$.
In practical terms, the infinite constraint could be hard to verify or apply directly, so often it's simplified by leveraging graph symmetry or orbit decomposition to reduce the condition to a finite, more manageable set of constraints.

    \begin{theo}\label{theo:arbitrary}
    Let $G$ be an arbitrary compatibility graph without self-loops, let $\mathcal{A}$ be a subgroup of $\mathrm{Aut}[G]$ and consider a greedy matching discipline.
        If the arrival probabilities satisfy the constraints in \eqref{eq:arbitrary_alpha} and if the discipline satisfies the conditions in \eqref{eq:arbitrary_mu}, then the Markov chain $\mathcal{M}(G,(\alpha_i)_{i=1,\cdots,n},GD)$ is strongly aggregable for the partition defined in \eqref{part1}, based on the subgroup $\mathcal{A}$ of $\mathrm{Aut}[G]$.
        
    \end{theo}
\begin{proof}
The proof follows that of Theorem \ref{theo:cycle}. 
Using the same notation, we consider
    $$Pr(m,{\cal{B}}_j) = \sum_{l\in V}\alpha_l \sum_{u \in {\cal{B}}_j } Pr({\rm arrival ~of~ letter} ~{\it l}~ {\rm provokes~
a~ transition~ from} ~{\it m}~ {\rm to} ~{\it u}).$$
By Lemma \ref{lemma:arbitrary} and Corollary \ref{corollary:arbitrary}, we find that the summations satisfy:
$$\begin{array} {l} 
\sum_{u \in {\cal{B}}_j } Pr({\rm arrival ~of~ letter} ~{\it l}~ {\rm provokes~
a~ transition~ from} ~{\it m}~ {\rm to} ~{\it u})  = \\
\\
\sum_{u \in {\cal{B}}_j } Pr({\rm arrival ~of~ letter} ~{\it \sigma(l) }~ {\rm provokes~
a~ transition~ from} ~{\it \Lambda_\sigma (m) }~ {\rm to} ~{\it u}). 
\end{array} $$

Clearly, if the arrival probabilities $\alpha_i$ satisfy the conditions in \eqref{eq:arbitrary_alpha}, it follows that $Pr(m,B_j)=Pr(p,B_j)$.
Thus, the lumpability condition is satisfied for an arbitrary graph $G$, completing the proof.
\end{proof}

\section{Extension to Non-greedy Matching Discipline}
In the previous sections, we only considered the greedy matching disciplines.
We now consider two families of non-greedy disciplines.
The analysis is simpler than above, because for each new discipline, we only need to check the two conditions of Corollary \ref{corollary:arbitrary} and then to apply Theorem \ref{theo:arbitrary}.

\subsection{Rejection Non-greedy Matching Disciplines}
We first consider a class of non-greedy matching disciplines which preserve the relation between the independent sets of the compatibility graph and the states of the system.
Our objective is to analyze the lumpability conditions for an arbitrary compatibility graph $G$ under this class of matching disciplines. 
We refer to this class of non-greedy matching disciplines as Rejection matching disciplines, which are defined as follows.
\begin{defin}[Rejection discipline]
    Let $m$ denote an arbitrary state, and let $x$ be an item arriving into the system.
    Under the non-greedy matching discipline Rejection, the system evolves as follows:
    \begin{itemize}
    \item If $\Delta(m,x)=\emptyset$: The item $x$ will be added to $m$ with probability $1$, yielding $m+e_x$ as the new state;
    \item If $\Delta(m,x)\neq\emptyset$: The item $x$ interacts with the state $m$ according to the following probabilistic rules:
    \begin{itemize}
        \item Rejection probability $\eta\in[0,1]$: The item $x$ is rejected and leaves the system even though it matches, and the state remains as $m$.
        \item Matching probability $(1-\eta)$: The item $x$ successfully matches with a compatibility item with type $u\in\Delta(m,x)$.
        Both $x$ and the matched item are removed from the system, resulting in a new state $m-e_u$.
        \end{itemize}
    \end{itemize}
\end{defin}
\begin{lemma}
    \label{lemma:rejection}
    Let $m\in\mathcal{S}$ be a state, and let $x\in V$ be an arriving item.
    Consider an arbitrary permutation $\sigma$ in subgroup $\mathcal{A}$, defining $p=\Lambda_\sigma(m)$ and $y=\sigma(x)$.
    Under the Rejection matching discipline, if the action of $x$ on $m$ leads to a state $r$, then the action of $y$ on $p$ leads to a state $s$.
    Furthermore, $r$ and $s$ are in the same macro-state.
\end{lemma}

We build on the analysis in Lemma \ref{lemma:arbitrary}, still considering two cases based on the behavior of item $x$ on state $m$.
Instead of repeating the proofs, we emphasize the changes in the transition probabilities induced by the non-greedy matching discipline, and the effect of the consistency conditions in equation \eqref{eq:remark} which is crucial for the lumpability condition.
\begin{proof}
For any state $m$ and an arriving item $x$, we have the following two cases:
\begin{itemize}
    \item Case 1: $\Delta(m,x)=\emptyset$. 
    This case is unaffected by rejection since no matching occurs. The item $x$ is added to $m$ with transition probability $1$, yielding $r=m+e_x$ as the new state.
    The transition probabilities remain invariant under automorphisms, preserving lumpability:
    $$\mu_{R(m,x)}(r)=\mu_{R(p,y)}(s)=1.$$
    Thus, the automorphism-based transition consistency condition in equation \eqref{eq:remark} simplifies to the consistency of arrival probabilities in equation \eqref{eq:arbitrary_alpha}, namely for all $\sigma\in\mathcal{A}$
    \begin{equation}
        \label{eq:arrival consistency}\alpha_x=\alpha_{\sigma(x)}=\alpha_y.
    \end{equation}
    \item Case 2: $W=\Delta(m,x)$ is not empty. 
    The item $x$ may attempt to match with an item $u\in W$, and one of the two outcomes occurs:
    \begin{itemize}
        \item Rejection: The system remains in state $m$, as item $x$ is rejected and no items are removed.
        The transition probability in this case is $\mu_{R(m,x)}(m)=\eta$.
        Considering that the rejection probability $\eta$ is a constant, which means that it is independent of the state.
        \item Match: The item matches with an item $u\in W$, and the transition leads to $r=m-e_{u}$, where $u$ is chosen uniformly from $W$, giving:
        $$\mu_{R(m,x)}(r)=(1-\eta)/|W|,$$ where $|W|$ is the size of compatibility set.
        As $\eta$ is a constant, it remains invariant under automorphisms, ensuring that transition probabilities remain consistent across symmetric states. Thus, the lumpability condition is preserved.
    \end{itemize}
\end{itemize}    
\end{proof}

\subsection{Generalizing the Rejection Probability}
Next, we extend the definition of rejection probability $\eta$ by allowing it to depend on both the system state and the arriving item.
We categorize this generalized rejection probability into three cases:
\begin{defin}
    The rejection probability function $\eta$ can be a function of the arriving item $x$, the current system state $m$, or both:
    \begin{enumerate}
    \item $\eta=h(x)$, where the probability of rejection is determined solely by the type of arriving item
    \item $\eta=h(m)$, where the probability of rejection depends on the current state.
    \item $\eta=h(m,x)$, where both factors influence the probability of rejection.
\end{enumerate}
\end{defin}
The intuitive example for item-dependent rejection probability $\eta=h(x)$ is that high-priority items are rarely rejected, while low-priority items are often rejected.
The state-dependent rejection probability function $\eta=h(m)$ is very useful in queueing or congestion models, where arriving items have a lower probability of being rejected for quick matching, regardless of the item type, as the system queue approaches capacity. 
The final case $\eta=h(m,x)$ is the most general case, and we can illustrate it with the following example:
    $$h(m, x)= \begin{cases}0.1, & x \in \text { high-priority }, \\ 0.7, & x \in \text { low-priority and }|m| \leq 2, \\ 0.3, & x \in \text { low-priority and }|m|>2 .\end{cases}$$
    Intuitively, high-priority items are always accepted regardless of queue size, while low-priority items are prone to be rejected when queues are small but accepted when queues are large.

As for the lumpability condition for arbitrary graphs in the above three cases, the system needs to satisfy the following two general conditions.

\begin{theo}
    \label{theo:symmetry invariance}
    Let $m\in\mathcal{S}$ be a state, and let $x\in V$ be an arriving item.
    For a rejection probability function $h(m,x)$, it satisfies the symmetry invariance condition under the automorphism subgroup $\mathcal{A}$ if and only if
    $$h(m,x)=h(\Lambda_\sigma(m),\sigma(x)),\quad \forall x\in\mathcal{V},m\in S,\sigma\in\mathcal{A}.$$
\end{theo}
This condition ensures that the rejection probability remains consistent across items in the same partitions under the automorphism subgroup $\mathcal{A}$, preserving the required equivalence between transitions.
Furthermore, when the rejection probability function depends solely on the arrival term, i.e., $h(x)$, or the state of the system, i.e., $h(m)$, the symmetry invariance condition can be specified as the following two lemmas.
\begin{lemma}
    If the rejection probability function depends only on the arriving item $x$, then the function satisfies the symmetry invariance condition under the automorphism subgroup $\mathcal{A}$ if and only if
    $$h(x)=h(\sigma(x)),\quad \forall x\in\mathcal{V},\sigma\in\mathcal{A}.$$
\end{lemma}
\begin{lemma}
    If the rejection probability function depends only on the current state $m$, then the function satisfies the symmetry invariance condition under the automorphism subgroup $\mathcal{A}$ if and only if
    $$h(m)=h(\Lambda_\sigma(m)),\quad \forall m\in\mathcal{S},\sigma\in\mathcal{A}.$$
\end{lemma}
Subsequently, when the arriving item $x$ is rejected, the state transition probability of remaining the original state is as follows:
\begin{equation*}
   \mu_{R(m,x)}(m)=h(m,x),\quad \mu_{R(p,y)}(p)=h(\Lambda_\sigma(m),\sigma(x)).     \end{equation*}
When the rejection probability function $h(m,x)$ satisfies the symmetry invariance condition in Theorem \ref{theo:symmetry invariance}, the state transition probabilities of both are equal. 
Further, under the arrival consistency in equation \eqref{eq:arrival consistency}, the lumpability condition holds in this case.
However, this condition does not control what happens when the arrival item is accepted and matched, leading to the next condition of the theorem.

\begin{theo}
\label{theo:transition consistency}
    Let $m\in\mathcal{S}$ be a state, and let $x\in V$ be an arriving item.
    Consider an arbitrary permutation $\sigma$ in subgroup $\mathcal{A}$, defining $p=\Lambda_\sigma(m)$ and $y=\sigma(x)$.
    If matching occurs (with probability $1-\eta$), the transition probability satisfies consistency under automorphisms if and only if the probability of selecting the item $u$ from $\Delta(m,x)$ depends only on an automorphism-invariant property of the system.
    Formally, for any $u$, the following equation must hold:
    $$\mu_{R(m,x)}(m-e_u)=\mu_{R(p,y)}(p-e_v), \quad \text{where }v=\sigma(u).$$
\end{theo}

To ensure that Theorem \ref{theo:transition consistency} holds, we now examine two specific matching rules that satisfy the required conditions.
The following two lemmas establish sufficient conditions under which the transition probability consistency is maintained.

\begin{lemma}
    Suppose that matching is chosen uniformly among the types of $\Delta(m,x)$, i.e., each item is selected with probability:
    $$\mu_{R(m,x)}(m-e_u)=\frac{(1-\eta)}{|\Delta(m,x)|}.$$
    Then, the transition probability satisfies the conditions required in Theorem \ref{theo:transition consistency}.
\end{lemma}

\begin{lemma}
    Suppose that matching is chosen based on the degree function $\deg_G(u)=|\Gamma(u)|$, where the probability of selecting an item $u\in\Delta(m,x)$ is proportional to its degree:
    $$\mu_{R(m,x)}(m-e_u)=\frac{\deg_G(u)}{\sum_{k\in\Delta(m,x)}\deg_G(k)}\times (1-\eta).$$
    The degree function naturally satisfies invariance under arbitrary automorphisms, i.e., $\deg_G(u)=\deg_G(\sigma(u))$.
    Therefore, the transition probability satisfies the conditions required in Theorem \ref{theo:transition consistency}.
\end{lemma}
Finally, we can derive the following corollary \ref{corollary:rejection} based on the above theorem related to the strong aggregation of the associated Markov chains under the non-greedy rejection matching discipline.
\begin{cor}
    \label{corollary:rejection}
    Consider a stochastic matching system over an arbitrary graph $G=(\mathcal{V},\mathcal{E})$ with a non-greedy rejection matching discipline $R$.
    If the rejection probability function $\eta$ satisfies the symmetry invariance condition in Theorem \ref{theo:symmetry invariance} and transition probability satisfies the consistency condition in Theorem \ref{theo:transition consistency} when matching occurs, the Markov chain $\mathcal{M}(G,(\alpha_i)_{i=1,\cdots,n},R)$ is strong aggregable based on the subgroup $\mathcal{A}$ for the partition defined in \eqref{part1}.
\end{cor}

\subsection{Threshold-Driven Non-greedy Matching Disiciplines}
Under the Rejection matching discipline, the states of the system remain independent sets due to the rejection mechanism.
We now extend our focus to a class of threshold-driven non-greedy matching disciplines which allow compatible items to coexist within the same states.
Therefore, the states are no longer independent sets anymore.

To proceed, we formally define the threshold-based non-greedy matching disciplines as follows:
\begin{defin}[Threshold-driven discipline]
    Let $m\in\mathcal{S}$ be a state, and let $x\in V$ be an arriving item.
    Given a threshold $\theta\in\mathbb{N}$ and a threshold function $h$, the system evolves according to the following rules:
    \begin{itemize}
        \item If $\Delta(m,x)=\emptyset$: the item $x$ is added to the state $m$ with probability $1$, yielding $m+e_x$ as the new state;
        \item If $\Delta(m,x)\neq \emptyset$: the interaction between $x$ and $m$ is determined by the value of the threshold function $h(m,x)$:
        \begin{itemize}
            \item If $h(m,x)>\theta$: the item $x$ matches with a compatible item in $\Delta(m,x)$, chosen according to uniform probability distribution.
            After matching, the state transitions to $m-e_u$, where $u\in\Delta(m,x)$ represents the matched item.
            \item If $h(m,x)\leq \theta$: the item $x$ is added to the state $m$ despite matching, resulting in a new state $m+e_x$.
        \end{itemize}
    \end{itemize}
\end{defin}
It is worth emphasizing that this framework enables adaptive matching policies, allowing the system to postpone decisions when immediate matching may not be optimal.

\begin{rem}
\label{remark:threshold function}
    The threshold function $h(m,x)$ determines the behavior of the matching discipline. 
    Below, we introduce three specific rules for $h(m,x)$:
\begin{enumerate}
    \item Global Queue Threshold ($h(m,x)=|m|)$: the threshold is based on the global queue size of the system.
    \item Local Density Threshold ($h(m,x)=|\Delta(m,x)|$): the threshold is determined by the number of compatible items in $\Gamma(x)$.
    \item Maximum Queue Length Threshold ($h(m,x)=\max_{u\in\Gamma(x)}m[u]$): the threshold depends on the maximum queue length of the compatible items in $\Gamma(x)$.
\end{enumerate}
\end{rem}
We will analyze each rule separately using lemmas to investigate the lumpability condition under these three specific rules. 

\begin{lemma}
    Let $m\in\mathcal{S}$ be a state, and let $x\in V$ be an arriving item.
    Consider an arbitrary permutation $\sigma$ in subgroup $\mathcal{A}$, defining $p=\Lambda_\sigma(m)$ and $y=\sigma(x)$.
    Under the threshold function $h(m,x)=|m|$, 
    if the action of $x$ on $m$ leads to a state $r$, then the action of $y$ on $p$ leads to a state $s$ with the same probability, and $r$ and $s$ are in the same macro-state.
\end{lemma}

\begin{proof}
The evolution of state $m$ upon the arrival of item $x$ in an arbitrary graph $G$ can be divided into three cases:
    \begin{itemize}
    \item Case 1:$\Delta(m,x)=\emptyset$. With probability $1$, the item $x$ is added to state $m$, and the state becomes $m+e_x$.
    By automorphism $\sigma$, if $p=\Lambda_\sigma(m)$ and $y=\sigma(x)$, it follows that:
    $$\mu_{R(m, x)}\left(m+e_x\right)=\mu_{R(p, y)}\left(p+e_y\right)=1.$$
    \item Case 2:$\Delta(m,x)\neq\emptyset$ and $|m|>\theta$.
    In this scenario, the set $\Delta(m,x)\neq\emptyset$ means that there exists at least one item $z\in\Gamma(x)$ such that $m[z]\neq 0$.
    By the automorphism property, we have
    $$(\sigma(x),\sigma(z))\in E\text{ and } \sigma(z)\in\Gamma(\sigma(x))\cap\Lambda_\sigma(m).$$
    This implies $\Gamma(y)\cap p\neq 0$.
    Consequently,
    the arrival of item $y$ also satisfies the condition:
    $$|p|=|\Lambda_\sigma(m)|=|m|>\theta.$$
    Since automorphisms preserve the structure of the graph and the number of items in each state.
    
    Combining the proof in Lemma \ref{lemma:arbitrary}, we have the following equations based on the automorphism property.
    $$|\Delta(m,x)|=q\text{ and }|\Delta(\Lambda_\sigma(m),\sigma(x))|=q,\quad\text{ for all }\sigma\in\mathcal{A}.$$
    Since the matched item is chosen uniformly from $\Delta(m,x)$ and $\Delta(p,y)$, we have $$\mu_{R(m,x)}(m-e_u)=\mu_{R(p,y)}(p-e_v),\quad v=\sigma(u).$$ 
    Therefore, we could satisfy the consistency conditions mentioned in equation \eqref{eq:remark} when we have $\alpha_x=\alpha_y$.
    \item Case 3: $\Delta(m,x)\neq\emptyset$ and $|m|\leq \theta$. 
    Same as in Case 2, we satisfy the following condition for item $y$:
    $$|p|=|\Lambda_\sigma(m)|=|m|\leq \theta.$$
    In this case, the item $x$ will be added to the state $m$, and the state becomes $m+e_x$. 
    Because of the property of automorphisms, we have
    $$\mu_{R(m,x)}(m+e_u)=\mu_{R(p,y)}(p+e_v)=1,\quad v=\sigma(u).$$
\end{itemize}
\end{proof}

Based on the proof of the above lemma, we will simplify the proof of the next two lemmas to avoid repetition.

\begin{lemma}
    Let $m\in\mathcal{S}$ be a state, and let $x\in V$ be an arriving item.
    Consider an arbitrary permutation $\sigma$ in subgroup $\mathcal{A}$, defining $p=\Lambda_\sigma(m)$ and $y=\sigma(x)$.
    Under the threshold function $h(m,x)=|\Delta(m,x)|$, 
    if the action of $x$ on $m$ leads to state $r$, then the action of $y$ on $p$ leads to 
    state $s$ with the same probability, and $r$ and $s$ are in the same macro-state.
\end{lemma}
\begin{proof}
The evolution of  state $m$ upon the arrival of item $x$ in an arbitrary graph $G$ can be divided into three cases:
    \begin{itemize}
    \item Case 1:$\Delta(m,x)=\emptyset$. The proof is exactly the same as the proof of the previous lemma.
    \item Case 2:$\Delta(m,x)\neq\emptyset$ and $|\Delta(m,x)|>\theta$.
    The difference with the previous proof is the following conditional statement:
    $$|\Delta(p,y)|=|\Delta(\Lambda_\sigma(m),\sigma(x))|$$
    Since automorphisms preserve the structure of the graph and the number of items in each state, we have:
    $$|\Delta(p,y)|=|\Delta(m,x)|>\theta.$$
    \item Case 3: $\Delta(m,x)\neq\emptyset$ and $|\Delta(m,x)|\leq \theta$. 
    Similarly, the difference lies in the conditional statement:
    $$|\Delta(p,y)|=|\Delta(m,x)|\leq \theta.$$
\end{itemize}
\end{proof}

\begin{lemma}
    Let $m\in\mathcal{S}$ be a state, and let $x\in V$ be an arriving item.
    Consider an arbitrary permutation $\sigma$ in subgroup $\mathcal{A}$, defining $p=\Lambda_\sigma(m)$ and $y=\sigma(x)$.
    Under the threshold function $h(m,x)=\max_{u\in\Gamma(x)}m[u]$, 
    if the action of $x$ on $m$ leads to state $r$, then the action of $y$ on $p$ leads to state $s$ with the same probability, and $r$ and $s$ are in the same macro-state.
\end{lemma}
\begin{proof}
The evolution of the state $m$ upon the arrival of item $x$ in an arbitrary graph $G$ can be divided into three cases:
    \begin{itemize}
    \item Case 1:$\Delta(m,x)=\emptyset$. The proof is exactly the same as the proof of the
previous lemma.
    \item Case 2:$\Delta(m,x)\neq\emptyset$ and $\max_{u\in\Gamma(x)}m[u]>S$.
    In this scenario, the arrival of item $y$ also satisfies the condition:
    $$\max _{v \in \Gamma(y)} p[v]=\max _{\sigma(u) \in \Gamma(\sigma(x))} p[\sigma(u)]=\max _{u \in \Gamma(x)} m[u]>S .$$
    \item Case 3: $\Gamma(x)\cap m\neq\emptyset$ and $\max_{u\in\Gamma(x)}m[u]\leq S$. 
    Same as in Case 2, we satisfy the following condition for item $y$:
    $$\max _{v \in \Gamma(y)} p[v]=\max _{\sigma(u) \in \Gamma(\sigma(x))} p[\sigma(u)]=\max _{u \in \Gamma(x)} m[u]\leq S .$$
\end{itemize}
\end{proof}

Based on the proofs of the above three lemmas, we can conclude that the lumpability condition holds in arbitrary graphs for all three threshold functions in Remark \ref{remark:threshold function} when $\alpha_i=\alpha$ for all $i\in\mathcal{V}$. 
From this, we can give the properties that need to be satisfied by the threshold function $h(m,x)$ in the following theorem in order for it to ensure that the associated Markov chain can satisfy the lumpability conditions.

\begin{theo}
    Let $G=(V,E)$ be a compatibility graph with a non-trivial automorphism group $\mathrm{Aut}[G]$.
    Consider a subgroup $\mathcal{A}\subseteq \mathrm{Aut}[G]$ and a threshold-driven non-greedy matching discipline $T$ defined by a threshold function $h(m,x)$.
    The associated Markov chain $\mathcal{M}(G,(\alpha_i)_{i=1,\cdots,n},T)$ is strongly aggregable if the threshold function $h(m,x)$ satisfies the following properties:
    \begin{enumerate}
        \item Symmetry Invariance:
        For any $\sigma\in\mathcal{A}$, let $p=\Lambda_\sigma(m)$ and $y=\sigma(x)$. Then,
        $$h(m,x)=h(p,y),\quad \forall m\in\mathcal{S},x\in V.$$
        
        \item Transition Consistency:
        For all $\sigma\in\mathcal{A}$, the transition probabilities are identical:
        $$\alpha_x\mu_{R(m,x)}(r_i)=\alpha_y\mu_{R(p,y)}(s_i),$$
        where $s_i=\Lambda_\sigma(r_i)$.

    \end{enumerate}
\end{theo}

\begin{proof}
The first property ensures that states within the same partitions behave consistently with respect to $h(m,x)$, leading to identical evaluations of the threshold condition.
The second property corresponds to the automorphism-based transition consistency described in Corollary \ref{corollary:arbitrary}.
Specifically, for the three cases induced by the arrival of an item $x$:
\begin{itemize}
    \item In Case 1 and Case 3, where no matching occurs, the transition probabilities satisfy:
    $$\mu_{R(m,x)}(m+e_x)=\mu_{R(p,y)}(p+e_y)=1.$$
    Here, the consistency of transition probabilities reduces to the consistency of arrival probabilities:
    $$\alpha_x=\alpha_y.$$
    \item In Case 2, where matching occurs, the selection of a compatible item in $\Gamma(x)\cap m$ is according to the uniform probability distribution.
    This ensures that for any $s_i=\Lambda_\sigma(r_i)$, the transition probabilities satisfy:
    $$\mu_{R(m,x)}(r_i)=\mu_{R(p,y)}(s_i).$$
\end{itemize}
\end{proof}

To summarize, we extend the concept of strong aggregation beyond traditional greedy matching disciplines by introducing non-greedy Rejection matching disciplines and Threshold-driven matching disciplines.
Our analysis demonstrates that lumpability conditions are preserved under these non-greedy disciplines, provided that the rejection function and threshold function satisfy certain properties.
These findings establish a theoretical foundation for studying more general stochastic matching models.

\section {Conclusion} 

In the sequel of this paper, we plan to check under which conditions the Markov chains are also exactly lumpable for the partitions we found here.  
Another idea is to study the chains associated with the multigraph compatibility graph. 
It is possible that the arguments based on the automorphism may also work for this topology and lead to efficient numerical algorithms. 
Clearly, it could be the case when all nodes have a self-loop because the Markov chain is finite. 
It may also be possible to find some structural description of the chains based on matrices (like in a QBD structure \cite{LaRa99}) or tensors (see for instance \cite{FPSt98,FPSt08}).

\bibliographystyle{unsrt} 
\bibliography{matching,Borne} 
\end{document}